\documentclass[11pt]{article}

\usepackage[margin=2.5cm]{geometry}
\usepackage{authblk}

\usepackage{amssymb,amsmath,amsthm,amsfonts,thmtools}
\usepackage{color}
\usepackage{algorithm}
\usepackage[noend]{algpseudocode}
\usepackage{yfonts}
\usepackage{mathrsfs}

\usepackage{graphicx}
\usepackage{xspace}

\usepackage{times}
\usepackage{verbatim}
\usepackage{enumitem}
\usepackage{graphicx,wrapfig,lipsum}
\usepackage[font=small]{caption}
\usepackage{subcaption}
\usepackage[english]{babel}
\usepackage[utf8]{inputenc}

\usepackage{url}
\usepackage[all]{xy}
\usepackage{rotating}
\usepackage{ifpdf}
\usepackage{tcolorbox}
\usepackage{lipsum}

\usepackage{mdframed}             
\usepackage{xifthen} 
\usepackage{mathtools}  

\DeclareMathAlphabet\mathbfcal{OMS}{cmsy}{b}{n}

\usepackage{comment}
\usepackage[T1]{fontenc}

\usepackage[hidelinks]{hyperref} 

\usepackage[section]{placeins}
\usepackage{float}
\usepackage[title]{appendix}
\usepackage{graphicx}
\usepackage{tikz}
\usetikzlibrary{automata,positioning}
\usepackage{dsfont}
\usepackage{xspace}
\usepackage{enumitem}
\setdescription{font=\normalfont}

\usepackage[normalem]{ulem}
\usepackage{xfrac}

\newcommand{\barQ}{\bar{Q}}  
\newcommand{\barR}{\bar{R}} 
\newcommand{\barS}{\bar{S}}

\newcommand{\calE}{\ensuremath{\mathcal E}\xspace}

\newcommand{\calQ}{\ensuremath{\mathcal Q}\xspace}

\newcommand{\tildeO}{{\tilde{O}}}

\newcommand{\tildeOmega}{{\tilde{\Omega}}}

\newtheorem{theorem}{Theorem}
\newtheorem{lemma}{Lemma}

\newtheorem{claim}[lemma]{Claim}

\newtheorem{observation}[lemma]{Observation}

\newcommand{\MultisetMAC}{{\textsf{MMAC}}}

\newcommand{\Flush}{{\textsc{Flush}}}

\newcommand{\RoundRobin}{{\textsc{RoundRobin}}}

\newcommand{\protP}{\mathbb{P}}
\newcommand{\TripleSel}{{\textsc{TripleSel}}}
\newcommand{\TripleSelFull}[3]{{\TripleSel}(#1,#2,#3)}
\newcommand{\Discharge}{{\textsc{Discharge}}}
\newcommand{\SparseQuasiGossip}{{\textsc{SparseQGossip}}}
\newcommand{\DeltaRegularQuasiGossip}{{\textsc{DeltaRegularQGossip}}}
\newcommand{\SimpleSparseQuasiGossip}{{\textsc{SimpleSparseQGossip}}}
\newcommand{\Condense}{{\textsc{Condense}}}
\newcommand{\InitCondense}{{\textsc{InitCondense}}}

\colorlet{darkgreen}{green!45!black}

\newcommand{\ignore}[1]{}

\newcommand{\margincomment}[2]%
{\marginpar{\footnotesize\raggedright {\color{red}#1}: #2}}

\newcommand{\etal}{et~al.}

\newcommand{\myparagraph}[1]{{\small\noindent\textbf{#1}}}
\newcommand{\myitparagraph}[1]{{\medskip\noindent\textit{#1}}}

\algnewcommand\algorithmiclet{\textbf{let}}
\algnewcommand\Let{\State \algorithmiclet\ }
\algnewcommand\algorithmiccase{\textbf{case}}
\algdef{SE}[CASE]{Case}{EndCase}[1]{\algorithmiccase\ #1}{\algorithmicend\ \algorithmiccase}%
\algtext*{EndCase}%
\algnewcommand{\IfThenElse}[3]{
  \State
  \algorithmicif\ #1\ \algorithmicthen\ #2\ \algorithmicelse\ #3}

\algnewcommand{\LeftComment}[1]{\Statex \hspace{-1em} \(\triangleright\) \emph{#1}}

\newcommand{\braced}[1]{{ \left\{ {#1} \right\} }}

\newcommand{\suchthat}{{\;:\;}}

\newcommand{\polylog}{\operatorname{polylog}}

\newcommand{\tinyminus}{{\scriptscriptstyle{-}}}
\newcommand{\tinyplus}{{\scriptscriptstyle{+}}}

\newcommand{\textmax}{{\textrm{max}}}

\newcommand{\inneighbors}{\textit{NB}^\tinyminus}
\newcommand{\outneighbors}{\textit{NB}^\tinyplus}
\newcommand{\indegree}{{\textrm{deg}^\tinyminus}}
\newcommand{\outdegree}{{\textrm{deg}^\tinyplus}}
\newcommand{\rumor}{\rho}
\newcommand{\packet}{\pi}

\newcommand{\aveindegree}{\Delta_{\scriptstyle\text{\normalfont ave}}}
\newcommand{\eligible}{\calE}

\newcommand{\pendingrumors}[2]{\textit{PR}_{#1}^{#2}}
\newcommand{\pendingpackets}[2]{\textit{PP}_{#1}^{#2}}

\newcommand{\SelLatency}[3]{L_{\scriptscriptstyle\text{\normalfont Sel}}(#1,#2,#3)}
\newcommand{\LayerSelLatency}{\Lambda}
\newcommand{\TSLatency}[3]{L_{\scriptscriptstyle\text{\normalfont{TrS}}}(#1,#2,#3)}
\newcommand{\LayerTSLatency}{\Gamma}

\newcommand{\flushtime}{{L_{\scriptscriptstyle\text{\normalfont Flush}}}}
\newcommand{\dischargetime}{{L_{\scriptscriptstyle\text{\normalfont Dis}}}}
\newcommand{\packetset}{\Pi}

\newcommand{\onethird}{{\frac{1}{3}}}
\newcommand{\txtonethird}{\textstyle{\frac{1}{3}}}

\newcommand{\sltwothirds}{{\sfrac{2}{3}}}

\newcommand{\half}{{\frac{1}{2}}}

\newcommand{\slhalf}{{\sfrac{1}{2}}}

\newcommand{\slthreehalves}{{\sfrac{3}{2}}}

\newcommand{\slfourthirds}{{\sfrac{4}{3}}}

\newcommand{\slthreefourths}{{\sfrac{3}{4}}}

\newcommand{\slfivefourths}{{\sfrac{5}{4}}}

\newcommand{\onefifth}{{\frac{1}{5}}}

\newcommand{\slthreefifths}{{\sfrac{3}{5}}}

\newcommand{\slelevensixths}{{\sfrac{11}{6}}}

\newcommand{\elevennineths}{{\frac{11}{9}}}

\newcommand{\slelevennineths}{{\sfrac{11}{9}}}

\newcommand{\ListLengths}{\setlength{\itemsep}{0ex}\setlength{\topsep}{1ex}\setlength{\partopsep}{0ex}}

\graphicspath{{./Figures/}}
\title{A Gossiping Protocol for Sparse Ad-Hoc Radio Networks\thanks{Research supported by NSF grant CCF-2153723.}}

\author[1]{Chao Wu}
\author[1]{Marek Chrobak}

\affil[1]{University of California at Riverside}

\begin{document}

\maketitle

\begin{abstract}
We study the problem of gossiping (all-to-all information exchange) in ad-hoc radio networks. Such a network is represented by a strongly-connected 
directed graph with $n$ vertices, whose topology is initially unknown to the protocol. 
In 2004, Gasieniec, Radzik, and Xin~\cite{Gasieniec_etal_det_gossip_04} gave an $\tildeO(n^\slfourthirds)$-time deterministic protocol for gossiping
in this type of networks, 
and closing the gap between their upper bound and the $\tildeOmega(n)$ lower bound on the time complexity of gossiping remains a central open problem. 
We develop a deterministic protocol for gossiping in ad-hoc radio networks that achieves running time
		$\tildeO((mn)^{\slthreefifths})$ for directed graphs with at most $m$ edges. 
		Our protocol improves on the $\tildeO(n^\slfourthirds)$ bound when $m = O(n^{c})$,
		for $c < \elevennineths$. We also present an $\tildeO(\Delta^{\slhalf}n)$-time gossiping
		protocol for $\Delta$-regular graphs.
\end{abstract}


\section{Introduction}
\label{sec: introduction}


Broadcasting and gossiping are fundamental primitives for information dissemination in networks.
In broadcasting, or one-to-all communication, a single node has a message 
that needs to be delivered to all nodes in the network.
In gossiping, or all-to-all communication, every node has its own unique
message, which we refer to as a \emph{rumor}, and the objective is to deliver each rumor to all nodes.

In this paper we study deterministic protocols for gossiping in \emph{ad-hoc radio networks}. 
We assume a basic model of such networks, as formalized in~\cite{Chlebus_etal_broadcast_02}, where
a network is modeled by a strongly-connected directed graph with $n$ vertices that
represent radio transmitters/receivers. Each vertex is identified by a unique label from the range 
$0,...,N-1$, for $N = O(n)$. These vertices communicate via transmissions over a radio channel
and the outgoing edges of a vertex represent its transmission range.
The term ``ad-hoc'' refers to the assumption that the topology of this graph is initially unknown.
There is no restriction on the size or
content of messages transmitted during the protocol. As all vertices use the same radio channel,
when two vertices transmit simultaneously signal collisions will occur at their shared
out-neighbors, and their messages will get lost.
The model does not assume any collision detection capability, so neither
the senders nor their shared out-neighbors will be notified when collisions occur.

Both broadcasting and gossiping in this model can be quite trivially solved in time $O(n^2)$ by
having only one vertex transmit at each time step, cyclically in a Round-Robin fashion,
thus avoiding collisions altogether. 
For broadcasting, this bound was improved in 2002 by Chlebus~\etal~\cite{Chlebus_etal_broadcast_02} to
$O(n^{\slelevensixths})$, and soon thereafter Chrobak~\etal~\cite{Chrobak_etal_fast_02} gave
a nearly-optimal protocol with running time $O(n\log^2n)$.
Further improvements in the poly-logarithmic factors were later achieved
by DeMarco~\cite{DeMarco_distributed_10} and by Czumaj and Davies~\cite{Czumaj_Davies_faster_16}, 
who provided upper bounds that are only a loglog-factor away from the lower bound of
$\Omega(n\log n)$~\cite{Bruschi_etal_lower_bound_97,Chlebus_etal_broadcast_02}.

In contrast, the complexity status of gossiping is still far from settled.
The first sub-quadratic protocol, with running time $\tildeO(n^{\slthreehalves})$\footnote{%
We express most of the time bounds for gossiping using the $\tildeO()$ and $\tildeOmega()$ notations,
the asymptotic notations that hide
$\polylog(n)$ factors. With the large gap between lower and upper bounds for gossiping, logarithmic
factors are of secondary importance. 
},
was given by Chrobak~\etal~\cite{Chrobak_etal_fast_02}. Xu~\cite{Xu_det_gossip_03}
reduced the poly-logarithmic factor in their upper bound.
A gossiping protocol that significantly improved this running time, to $\tildeO(n^{\slfourthirds})$,
was developed by Gasieniec~\etal~\cite{Gasieniec_etal_det_gossip_04} in 2004, and it is still
the state-of-the-art, except
for a minor improvement of the polylogarithmic factor in~\cite{kuschner2024permutation}.
Closing or reducing the remaining gap between the upper and lower bounds of
$\tildeO(n^{\slfourthirds})$ and $\tildeOmega(n)$, respectively, remains a major open problem in this area.

These upper bounds can be refined by expressing them in terms of additional network attributes, for example
the graph's diameter $D$ or its maximum vertex in-degree $\Delta$. 
Examples of such results include the $\tildeO(D\Delta^2)$-time protocol by Clementi~\etal~\cite{Clementi_etal_01},
later improved to $\tildeO(D\Delta^{\slthreehalves})$ by Gasieniec and Lingas~\cite{gasieniec2002adaptive},
or the $O(\Delta n)$-time protocol by Gasieniec~\etal~\cite{Gasieniec_etal_det_gossip_04}.


\myparagraph{Our contributions.}
We study gossiping in \emph{sparse} radio networks,
where the number of edges is sub-quadratic. The main result of this paper is
the following theorem.

\begin{theorem}\label{thm: improved bound sparse graphs}
There is a deterministic protocol for ad-hoc radio networks
that completes gossiping in time $\tildeO((mn)^{\slthreefifths})$ for directed graphs with
$n$ vertices and at most $m$ edges.
\end{theorem}

This protocol is faster than the $\tildeO(n^{\slfourthirds})$-time protocol in~\cite{Gasieniec_etal_det_gossip_04}
as long as  $m = O(n^{c})$, for $c < \slelevennineths$.
Unlike other refinements of the result in~\cite{Gasieniec_etal_det_gossip_04}, we do not assume a uniform bound on vertex degrees.  
The proof of Theorem~\ref{thm: improved bound sparse graphs} is given in
Section~\ref{sec: faster protocol for sparse graphs}. A key ingredient in achieving
this running time is a novel protocol (see Section~\ref{sec: auxiliary mac protocol}) for a variant of the
multiple-access channel problem in which packets to be
transmitted on the channel may have multiple copies in different competing processors.

We also consider the case of \emph{$\Delta$-regular graphs}, 
in which every vertex has both its in-degree and out-degree equal to $\Delta$. 
As mentioned earlier, if only in-degrees are upper bounded by $\Delta$, gossiping can be
solved in time $O(\Delta n)$~\cite{Gasieniec_etal_det_gossip_04}.
In Section~\ref{sec: delta-regular graphs} we improve it if the graph is $\Delta$-regular,
by proving the following theorem:

\begin{theorem}\label{thm: Delta-regular graphs}
There is a deterministic protocol for ad-hoc radio networks
that completes gossiping in time $\tildeO(\Delta^{\slhalf} n)$ for $\Delta$-regular graphs with  $n$ vertices.
\end{theorem}


\myparagraph{Other related work.}
As in other settings involving uncertainty, the running time of protocols in ad-hoc radio
networks can be improved, in expectation, by using randomization. For broadcasting,
the randomized case is completely solved, with the optimal running time known to be $O(D\log(n/D)+\log^2n)$~\cite{alon_etal_lower_bound_91,Kushilevitz_Mansour_lower_bound_98,Czumaj_Rytter_broadcast_06}. 
For gossiping, randomized protocols in~\cite{Czumaj_Rytter_broadcast_06,Liu_Prabhakaran_randomized_02,Chrobak_etal_randomized_04}
are only within a small poly-logarithmic factor from the 
lower bound~\cite{Kushilevitz_Mansour_lower_bound_98,Liu_Prabhakaran_randomized_02}.

Relatively little is known about time complexity of other communication primitives
in ad-hoc radio networks, for example information gathering%
\footnote{Information gathering is not
a special case of gossiping, because it does not assume strong connectivity.}
(all-to-one communication)~\cite{Chrobak_etal_information_gathering_2021}
or multi-broadcasting (many-to-all communication)~\cite{Clementi_dist_bcast_03}.

There are multiple other variants of the radio network model in the literature.
For instance, one important model generalizes the one from~\cite{Chlebus_etal_broadcast_02} by allowing
a polynomial-size label range (that is, $N = O(n^c)$, for a constant $c\ge 1$). In particular,
the $\tildeO(n^{\slfourthirds})$-time gossiping protocol in~\cite{Gasieniec_etal_det_gossip_04}
applies to this model.

The case of undirected graphs has been relatively well studied.
(See~\cite{Gasieniec_etal_large_labels_2007}, for example.)
Protocols in this model are generally faster, partly because bidirectional communication
facilitates direct feedback, thus helping to decrease the number of retransmissions and in turn reducing congestion.
For example, if spontaneous transmissions are allowed, it is possible to speed-up the process of broadcasting for undirected graphs~\cite{czumaj_davies_spontaneous_broadcasting_2021}.
Other variants in the literature include
models with restrictions on message size~\cite{Christersson_etal_gossiping_bounded_2002},
some forms of collision detection~\cite{Chrobak_etal_information_gathering_2021},
different assumptions about the range of vertex labels~\cite{Gasieniec_etal_large_labels_2007}, or
knowledge about graph topology~\cite{gkasieniec2007time}.

Interested readers can find more information about radio network protocols in the survey
papers by Gasieniec~\cite{Gasieniec_deterministic_broadcasting_2016},
Peleg~\cite{Peleg_broadcasting_review_2007},
and Itai~\cite{Itai_randomized_broadcasting_encyclopedia_2016}. 



\section{Preliminaries}
\label{sec: preliminaries}


\myparagraph{Ad-hoc radio networks.}
We use a basic model from~\cite{Chlebus_etal_broadcast_02}, where
a radio network is defined as a strongly-connected directed graph $G = (V,E)$, whose vertices represent 
transmitters/receivers
and edges represent their ranges: $(u,v)\in E$ iff $v$ is within the range of the radio signal from $u$.
The number of vertices is $n = |V|$ and the number of edges is $m = |E|$. For $v\in V$,
by $\inneighbors(v)$ and $\outneighbors(v)$ we denote the sets of in-neighbors and out-neighbors of $v$.
The in-degree and out-degree of $v$ are $\indegree(v) = |\inneighbors(v)|$ and $\outdegree(v) = |\outneighbors(v)|$.

The formal model involves the following assumptions:
\begin{itemize}
\item The vertices have unique labels from the range $[N] = \braced{0,1,...,N-1}$, for $N = O(n)$.
\item The graph topology is not revealed when a protocol starts. We only assume that the 
protocol knows $N$ and an upper bound $m$ on the number of edges.  

\item Initially, each vertex $v\in V$ possesses a unique message $\rumor_v$ that we refer to as
its \emph{rumor}. 
\item The time is discrete, with time steps numbered $0,1,2,...$. 
At each step, a vertex can either listen to transmissions from other vertices
or it can transmit one message.
\item
There is no restriction on the size or content of messages. 
\item  When a vertex $u$ transmits in some time step, 
its message is sent to each of its out-neighbors $v \in \outneighbors(u)$. 
If $u$ is the only in-neighbor of $v$ transmitting at this step then $v$
will receive $u$'s message. 
But if two or more in-neighbors of $v$
transmit at the same time step, a \emph{collision} occurs at $v$. Then 
$v$ does not receive any message. The model does not assume any
collision detection capability, so the senders do not know that their messages
were not received by $v$, and $v$ does not know about their transmission attempts.
\end{itemize}

The goal of gossiping is to \emph{disseminate} each rumor,
where a rumor is considered disseminated if it's delivered to all vertices in the network. 
A gossiping protocol is defined by specifying the actions of each vertex $v\in V$
at each time step, based on $v$'s label and on 
the record of all communications received by $v$ before the current step.
The running time of a protocol is the number of steps until the gossiping task is completed.

To simplify presentation, we assume that $N = n$ (this does not affect our proofs).
We also allow vertices to simultaneously receive and transmit. By standard arguments,
a protocol can be modified to perform only one of these actions at each step,
without affecting its asymptotic running time.
Since there is no restriction on the message size, when a vertex transmits, it can as well transmit
the complete history of its communications so far. In particular,
a message may contain multiple rumors.
We also assume that each vertex knows
its in-degree, as this can be accomplished with a single round of $\RoundRobin$ (see below).


\myparagraph{Selectors.}
Selectors are combinatorial structures whose uses include congestion control in multiple-access channels. 
Formal definitions of selectors in the literature
vary (e.g.~\cite{Erdos_etal_families_85,komlos_greenberg_non-adaptive_mac_1985,DeBonis_etal_selectors_groups_2003}), 
but all are fundamentally equivalent to the definition below.

Given two sets $Q, X\subseteq [n]$ and $z\in X$, we say that $Q$ \emph{isolates $z$ from $X$}
if $Q\cap X = \braced{z}$. 
A family $\calQ$ of sets is said to \emph{isolate $z$ from $X$} if at least one set in $\calQ$ does.

Let $1\le k\le n$, and
let  $\barS = (S_0, S_1,..., S_{\ell-1})$ be a sequence of sets, where $S_i\subseteq [n]$ for each $i$.
We will call $\barS$ an \emph{$(n,k,g)$-selector} iff
the family of sets $\braced{S_0,...,S_{\ell-1}}$ isolates at least $g$ elements from 
each $k$-element subset of $[n]$. We will refer to this $\ell$ as the \emph{latency} of $\barS$.
The construction by De Bonis~{\etal}~\cite{DeBonis_etal_selectors_groups_2003}
implies the existence of an $(n,k,g)$-selector with latency $\ell = \SelLatency{n}{k}{g} = O(k^2\log ( n )/(k - g + 1))$. 

In Section~\ref{sec: auxiliary mac protocol} 
we will use the following property of $(n,k,g)$-selectors, that follows
directly from their definition:
for each set $X$ with $|X|=k$ and each non-empty subset $Z\subseteq X$ with $|Z| \ge k-g+1$,
a $(n,k,g)$-selector isolates from $X$ at least one element of $Z$.


\myparagraph{Basic protocols.}
A simple protocol for gossiping called $\RoundRobin$ has
vertices labelled $0,1,...,n-1$ transmit one by one. There are no collisions, so each
vertex will successfully transmit, and thus
after $n-1$ rounds of $\RoundRobin$ each rumor will be disseminated, completing gossiping in $O(n^2)$ time.

Faster gossiping protocols in the literature use selectors. When we say that a protocol
\emph{executes or runs an $(n,k,g)$-selector $\barS$} we mean that at each time $t$ it uses $S_t$ to determine 
which vertices transmit: a vertex $u$ transmits at time $t$ if and only if $u\in S_t$.
(Here, we identify $u$ by its label.)
Consider a vertex $v$ with $\indegree(v)\le k$.
With an $(n,k,k)$-selector
it will take time $\SelLatency{n}{k}{k} = \tildeO(k^2)$ for $v$ to receive messages from all its in-neighbors,
which, for $k = o(\sqrt{n})$, is faster than $\RoundRobin$'s time $n$.


\myparagraph{Flushing.}
In our protocols, at each step of the execution all rumors will be designated as either \emph{pending} or \emph{flushed}. 
Initially all rumors are pending.
Each vertex $v$ will maintain its own list $\pendingrumors{v}{}$ of pending rumors it received, where initially
$\pendingrumors{v}{} = \braced{\rumor_v}$.  As $v$ receives messages with new rumors, they
are added to $\pendingrumors{v}{}$. Rumors lose their pending status, and are removed from each set $\pendingrumors{v}{}$,
after they are ``flushed'' from the network as a result of completing a procedure $\Flush()$, described below.

Procedure $\Flush(u)$, for a vertex $u$, is simply an $\tildeO(n)$-time 
protocol that disseminates all rumors in $\pendingrumors{u}{}$ by broadcasting them in one message.
As soon as $\Flush(u)$ completes,
the status of rumors in $\pendingrumors{u}{}$ changes from pending to flushed, and at this time
each vertex $v$ sets $\pendingrumors{v}{} = \pendingrumors{v}{} \setminus\pendingrumors{u}{}$.  We remark that
it is possible for $\rumor_v$ to be disseminated even if it's not flushed, in which case $\rho_v$'s status remains pending.

Gossiping protocols in the literature typically
involve spreading rumors across the network, and flushing rumors from vertices that collected many pending rumors.
In our protocols, the conditions for flushing will be more subtle.
A vertex $v$ will be called \emph{flush-eligible} if it satisfies a predicate $\eligible(v)$ that is a parameter of the protocol.
Protocol~$\Discharge(\eligible)$ shown below combines the process of finding a flush-eligible vertex and
the process of flushing its pending rumors. In this protocol, $\flushtime$ denotes the running time of procedure $\Flush()$.

\begin{center}
\framebox{
\begin{minipage}{3in}
\begin{tabbing}
aa \= aa \= aa \= aa \= aa \= aa \= \kill
\textbf{Protocol} $\Discharge(\eligible)$ \\
\> find $v\in V$ such that $\eligible(v)$ \\
\> \textbf{if} $v$ exists \textbf{then} execute $\Flush(v)$\\
\>\> \textbf{else} do nothing for time $\flushtime$
\end{tabbing}
\end{minipage}
}
\end{center}

One simple example of a flush-eligibility predicate is $\eligible(v) \equiv (|\pendingrumors{v}{}|\ge \beta)$,
for some suitable threshold $\beta$.
As shown in~\cite{Chrobak_etal_fast_02}, finding a $v$ with maximum load $|\pendingrumors{v}{}|$ can be done in  
time $\tildeO(n)$, so such a flush-eligible vertex can be found in time $\tildeO(n)$. This running time remains the same
as long as the flush-eligibility predicate uses only local information from a vertex, which is the case
in our protocols.  Thus, denoting by $\dischargetime$ the running time of $\Discharge()$,
throughout the paper we assume that $\dischargetime = \tildeO(n)$.


\myparagraph{Quasi-gossiping.}
We use the concept of \emph{quasi-gossiping} from~\cite{Gasieniec_etal_det_gossip_04},
which relaxes the objective of gossiping by requiring that each rumor only gets
\emph{quasi-disseminated}, which means that it
will be either disseminated or delivered to a vertex whose rumor was disseminated.
Any quasi-gossiping protocol $\protP$ can be converted into a gossiping protocol with
the same asymptotic running time as follows:
first run $\protP$, then repeat blindly the sequence of transmissions from that run of $\protP$.
(Caution: this is not the same as running $\protP$ again.)
In this protocol, if a rumor $\rumor_u$ was not disseminated by $\protP$, it must have
been delivered to some vertex $v$ whose $\rumor_v$ was disseminated by $\protP$.
Since in the second part of the protocol we execute exactly the same transmissions,
these transmission will result in disseminating $\rumor_u$ from $v$.


\myparagraph{Multiple transmission channels.}
It is convenient to present protocols as using multiple radio channels,
and allowing them to listen and transmit to all channels in a single step.
Collisions involve only transmissions on the same channel.
A protocol with $\mu$ channels can be simulated on one channel by
dividing the time into blocks of length $\mu$, and simulating one step
with $\mu$ channels by executing the same actions, in any order, in one block
(and ignoring all new information until the block ends).
This increases the running time by a factor of $\mu$, so, as long as 
$\mu = O(\polylog(n))$, an $\tildeO$-estimate on the running time is not affected.


\myparagraph{Rounding convention.}
To avoid clutter, in our formulas for various integer parameters and running times we will
omit rounding. All these values can be rounded up or down, as needed, without affecting the
asymptotic analysis. Some quantities, in particular 
running times of various sub-protocols used in the paper, will be further rounded to powers of two.


\section{Warmup: A Simpler Protocol}
\label{sec: simpler protocol}


To illustrate some key ideas in our protocol, we first show a simpler protocol
with running time $\tildeO(n^{\slthreefourths} m^{\slhalf})$.
As discussed in Section~\ref{sec: preliminaries}, it is sufficient to only show how to accomplish quasi-gossiping.
We can also assume that a protocol can use $O(\log n)$ radio channels and that each vertex knows its in-degree.

The protocol divides all vertices of degree at most $\sqrt{n}$ into at most $\half\log n$ layers, 
where each layer consists of vertices whose in-degrees are between two consecutive powers of $2$. 
Each group will have a 
corresponding selector to process transmissions into vertices in this group. We will present the protocol as
using different channels for different selectors. 
Vertices with degree larger than $\sqrt{n}$ form another group, that will be handled using $\RoundRobin$.
Each group has a different threshold value for flush-eligibility, with higher degrees assigned lower threshold values.
The key idea is this: Most vertices have in-degrees close to the average in-degree $\aveindegree = m/n$.
Vertices whose in-degree is much larger than $\aveindegree$ 
need a selector with large latency, slowing down the progress of
gossiping, but there are also fewer of them, so their rumors require fewer flush operations.

Let $l = \log \aveindegree$ and $h = \half\log n$.
We partition the set of vertices into $h-l+1$ layers, based on their in-degree:
$
V_l = \braced{v\in V\suchthat 1 \le \indegree(v) \le 2^l}$,
$V_j =\braced{v\in V\suchthat 2^{j-1} < \indegree(v)  \le 2^{j}}$  for $j = l+1,...,h-1$, and
$V_h = \braced{v\in V\suchthat 2^{h-1} < \indegree(v)  < n}$.
Using standard counting (an analogue of Markov's inequality) we obtain:
\begin{observation}\label{obs: bound on Vj simple}
$|V_j|\le n/2^{j-1-l}$ for $j=l,...,h$.
\end{observation}

For each $v\in V$ its list $\pendingrumors{v}{}$ of pending rumors is
partitioned into subsets according to
degrees: $\pendingrumors{v}{j} = \pendingrumors{v}{}\cap \rho(V_j)$, for all $j = l,...,h$.
(Naturally, $\rumor(U) = \braced{\rumor_v\suchthat v\in U}$ denotes
the rumors that originate from a set $U\subseteq V$.)
For $j=l,...,h-1$,	let $\LayerSelLatency_{j} = \SelLatency{n}{2^j}{2^j} = O(2^{2j}\log n)$ denote
the latency of the $(n,2^j,2^j)$-selector.
For the special case $j=h$ we will let $\LayerSelLatency_{h} = n$,
namely the latency of $\RoundRobin$. Also, for all $j$ we let $\beta_j = n / 2^{(3j-l)/2}$.

In our protocol, a vertex $v$ is considered \emph{flush-eligible} if it satisfies the
following flush-eligibility predicate:
$\eligible_0(v) \;\equiv\; 
	(\,\exists j \in \braced{l,...,h} \suchthat |\pendingrumors{v}{j}|\ge \beta_j\,)$.

\begin{center}
	\framebox{
		\begin{minipage}{5in}
			\begin{tabbing}
				aa \= aa \= aa \= aa \= aa \= aa \= \kill
				\textbf{Protocol} {\SimpleSparseQuasiGossip} \\
				\> Let $A =  \sum_{j=l}^h \beta_j\LayerSelLatency_{j}$ \,,\,  $F = \sum_{j=l}^h |V_j|/\beta_j$\,,\, and \, $T = A + \dischargetime F$\\
				\> \textbf{In parallel} execute the following processes for time $T$: \\
				\>\> \textit{Channel $j$, for $j=l,\ldots,h-1$:} Repeatedly run the $(n,2^j,2^j)$-selector \\
				\>\> \textit{RR~Channel:} Repeatedly run RoundRobin \\
				\>\> \textit{Flush~Channel:} Repeatedly execute $\Discharge(\eligible_0)$
			\end{tabbing}
		\end{minipage}
	}
\end{center}

Intuitively, $\dischargetime F$ bounds the
time needed by all flushes, and $A$ upper bounds the time needed for a rumor to traverse
a path with at most $\beta_j$ vertices in $V_j$, for each $j$. 


\myparagraph{Running time.}
As the number of radio channels is $h-l+2 = O(\log n)$, the running time of Protocol~{\SimpleSparseQuasiGossip} is $\tildeO(T)$, 
so it is sufficient to show that $T = \tildeO(n^{\slthreefourths} m^{\slhalf})$.

For each $j< h$, we have 
$\beta_j\LayerSelLatency_{j} = n\cdot 2^{(j+l)/2} 
							\le n\cdot 2^{(h+l)/2} 
							= \tildeO(n^{\slfivefourths}\aveindegree^{\slhalf}) 
							= \tildeO(n^{\slthreefourths} m^{\slhalf})$.
For $j = h$, $\beta_h\LayerSelLatency_{h} = \tildeO(n^{\slthreefourths} m^{\slhalf})$ as well. 
Since $h\le\log n$, we obtain that $A = \tildeO(n^{\slthreefourths} m^{\slhalf})$.
Also, $|V_j|/\beta_j \le 2^{(j+l)/2} \le 2^{(h+l)/2}$ for each $j$, and $\dischargetime = \tildeO(n)$,
so $\dischargetime F = \tildeO(n^{\slthreefourths} m^{\slhalf})$ as well. Therefore $T = \tildeO(n^{\slthreefourths} m^{\slhalf})$.


\myparagraph{Correctness.}
It remains to prove correctness, namely that after time $T$ the protocol will complete quasi-gossiping.
We start with two observations.

\begin{observation}\label{obs: simple prot, number of flushes}
If the protocol executes $F$ flushes then gossiping will be completed. 
\end{observation}

This is quite trivial: If $\Discharge()$ executes a flush from a vertex $v$ with
$|\pendingrumors{v}{j}|\ge \beta_j$, the number of pending rumors from $V_j$ decreases by $\beta_j$.
So for this $j$ we can have at most $|V_j|/\beta_j$ flushes, and the observation follows
by adding these bounds over all $j = l,...,h$.

\begin{observation}\label{obs: simple prot, path traveled}
For any vertex $z$, at time $A$ either its
rumor $\rumor_z$ has been flushed, or it has traversed each path that starts at $z$ and has 
at most $\beta_j$ vertices from $V_j$, for each $j = l,...,h$.
\end{observation}

This is true because, by the definition of selectors and $\RoundRobin$, once $\rumor_z$ reaches a vertex $x$, 
the time for it to traverse an edge  $(x,y)$, where $y\in V_j$ is bounded by $\LayerSelLatency_j$. 
Summing these costs along the path gives us time at most $A$.

\smallskip
Now, fix some arbitrary vertex $v$.
To show correctness we need to argue that after $T$ steps $v$'s rumor $\rumor_v$
will be quasi-disseminated. 

If $\rumor_v$ was quasi-disseminated by time $A$ then we are done, so suppose it's not.
Let $w$ be a vertex closest to $v$ that has not received $\rumor_v$ by time $A$.
Then there is a vertex $u\in \inneighbors(w)$ that has received $\rumor_v$.
As $\rho_v$ was not flushed, $\rumor_v$ must have traveled some path $P$ from $v$ to $u$
as a result of applying selectors on Channels~$l,...,h-1$ or $\RoundRobin$ on Channel~$h$.
Since $w$ has not received $\rumor_v$, Observation~\ref{obs: simple prot, path traveled} implies that
there is $j$ for which path $Pw$ must have strictly more than $\beta_j$ vertices from $V_j$.
Then $P$ has at least $\beta_j$ vertices from $V_j$, and their rumors (which are not yet flushed,
because $\rumor_v$ is not quasi-disseminated)
traveled with $\rho_v$ along $P$. So at time $A$ we have $\rumor(V_j\cap P)\subseteq \pendingrumors{u}{j}$,
where $|\rumor(V_j\cap P)|\ge \beta_j$, making $u$ flush-eligible.
During the remaining $\dischargetime F$ steps $A+1,...,T$, the protocol executes $F$ discharge operations.
These operations will execute flushes as long as $u$ remains flush-eligible.
If they all execute flushes, Observation~\ref{obs: simple prot, number of flushes}
implies that $\rumor_v$ will be disseminated.
If not, $u$ must cease being flush-eligible before time $T$, so
at least one of the rumors from $V_j\cap P$ must have been flushed before this time.
This means that $\rumor_v$ was delivered to a vertex whose rumor was disseminated, completing the proof of correctness.


\section{An Auxiliary MAC Protocol}
\label{sec: auxiliary mac protocol}


We now introduce a variant of the multiple-access channel problem,
called $\MultisetMAC$, that will be used in our protocol in Section~\ref{sec: faster protocol for sparse graphs}.
In an instance of $\MultisetMAC$ we are given two sets $X , \packetset \subseteq [n]$,
where $|X| = k$. The elements of $X$ represent labels of $k$ processors
and the elements of $ \packetset$ represent labels of \emph{packets}.
(From now on, we refer to processors and packets by their labels.)
Each processor $p$ may be storing a set of packets, that we denote by $\packetset_p \subseteq  \packetset$. 
The processors have access to a common transmission channel. At each step a processor $p$
may be either idle or it may transmit its packet set $ \packetset_p$ on the channel. For the
transmission to succeed, $p$ would have to be the only processor that transmits at this
time step. A protocol for $\MultisetMAC$ specifies, for each processor $p$, at what
time steps $p$ transmits. There is no feedback from the channel, so a protocol must be \emph{oblivious}: for 
each processor $p$ its transmission times depend only on $p$ and on its packet set $ \packetset_p$. 
The goal is to design an oblivious protocol that quickly transmits all packets held by the processors, that is $\bigcup_{p\in X}  \packetset_p$.

The $\MultisetMAC$ problem can be solved in time $\tildeO(k^2)$ by applying a
selector based on the processor labels, or in time $O(n)$ using $\RoundRobin$. We show that if the number of packets to transmit is 
small then it is possible to do better. 
The basic idea is to take advantage
of redundancy: if several processors have the same packet then
only one of them needs to succeed for this packet to reach the channel.
Our protocol combines three different selectors, one of them applied to
rumor labels rather than processor labels.
%
%
\begin{center}
\framebox{
\begin{minipage}{5.3in}
\textbf{Protocol}~$\TripleSelFull{n}{k}{r}$. 
Let $q = k/r$ (we assume that $r \leq k$). We use three selectors:
an $(n,k,k-q+1)$-selector $\barS$, an $(n,r,r)$-selector $\barR$, and an $(n,q,1)$-selector $\barQ$.

\noindent
\textbf{In-parallel} each processor $p$ executes the following transmissions:
\begin{description}
\item{\emph{Channel~1:}} Execute $\barS$.
\item{\emph{Channel~2:}} The time is divided into \emph{blocks}, each of length $\SelLatency{n}{q}{1}$.
	When the $j$-th block starts, it is designated as \textit{active} iff $R_j\cap  \packetset_p \neq \emptyset$.
	In this block do this: if it's active, execute $\barQ$, otherwise do not transmit at all.
\end{description}
\end{minipage}
}
\end{center}


\myparagraph{Analysis.}
Let $\packetset' = \bigcup_{p\in X} \packetset_p$ be the set of packets held by all processors.
We claim that if $|\packetset'| \leq r$
then protocol $\TripleSel()$ will transmit all packets from $\packetset'$ in time $\tildeO(kr)$. 
To prove this, consider a packet $\pi \in  \packetset'$, 
and let $X_\pi$ be the set of processors $p$ with $\pi\in  \packetset_p$. 

We consider two cases.
If $|X_\pi|\ge q$ then, as explained in Section~\ref{sec: preliminaries}, selector $\barS$ will isolate from $X$ at least one
element from $X_\pi$ in at most $\SelLatency{n}{k}{k-q+1} = \tildeO(k^2/q)$ steps, so
$\pi$ will get transmitted on Channel~1 in time at most $\tildeO(k^2/q) = \tildeO(kr)$.
Next, consider the case when $|X_\pi| < q$. 
By the definition of $\barR$, there is $j = \tildeO(r^2)$ for which $R_j\cap \packetset' = \braced{\pi}$.
Then $R_j\cap \packetset_p = \braced{\pi}$ for all $p\in X_\pi$ and $R_j\cap \packetset_p = \emptyset$ for $p\in X\setminus X_\pi$.
Thus all processors $p\in X_\pi$ will have their block $j$ active on Channel~2 and no other processors will. Then, by
the definition of $\barQ$, one of the processors in $X_\pi$ will get isolated from $X_\pi$ in at least one of the $\SelLatency{n}{q}{1}$
steps of this block, and in this  step it will successfully transmit. 
Summarizing, this processor will succeed in transmitting $\pi$ on Channel~2 in time 
$j\cdot \SelLatency{n}{q}{1} = \tildeO(r^2q)  = \tildeO(kr)$.

Thus, denoting by $\TSLatency{n}{k}{r}$ the running time of $\TripleSelFull{n}{k}{r}$, we obtain:

\begin{lemma}\label{lem: TMAC protocol}
If $|\bigcup_{p\in X} \packetset_p|  \leq r$ then
Protocol~$\TripleSelFull{n}{k}{r}$ transmits all packets from $\bigcup_{p\in X} \packetset_p$
in time $\TSLatency{n}{k}{r} = \tildeO(kr)$.
\end{lemma}


\section{Faster Protocol for Sparse Graphs}
\label{sec: faster protocol for sparse graphs}
In this section, we present our main result, a gossiping protocol with running time $\tildeO((mn)^{\slthreefifths})$,
thus proving Theorem~\ref{thm: improved bound sparse graphs}. As explained in Section~\ref{sec: preliminaries}, 
it is sufficient to only present the quasi-gossiping part of our protocol, and we can allow it to use $O(\log n)$ radio channels.


\myitparagraph{Overview.}
One idea behind our protocol was already used in Protocol~{\SimpleSparseQuasiGossip} in Section~\ref{sec: simpler protocol},
namely that the relative frequency of vertices with degrees above a given threshold
is inversely proportional to this threshold, and this lower frequency of high-degree vertices
compensates for the large delay to traverse them. We again partition vertices into those of ``low'' and ``high'' degree,
although the threshold that separates them is not $\sqrt{n}$ anymore. Low-degree vertices are processed essentially in
the same way as in Protocol~{\SimpleSparseQuasiGossip}, but for high-degree vertices our new protocol
will proceed in phases, gradually ``condensing''
the set of not-yet quasi-disseminated rumors, namely grouping them in vertices of higher and higher degrees,
until eventually, after $O(\log n)$ phases, all rumors will be quasi-disseminated. 
This grouping will be achieved by procedures called Procedure~$\InitCondense()$ and $\Condense()$ that we will describe shortly.
The speed-up in transmissions into high-degree vertices is obtained by applying Protocol~$\TripleSel()$ from the previous section,
instead of selectors. (The ``channel'' in Protocol~$\TripleSel()$ will be a high-degree vertex, and the competing processors will be
its in-neighbors.)
This needs to be done with care, because Protocol~$\TripleSel()$ is not a simple
black-box substitute for selectors. Namely, the speed-up described in 
Lemma~\ref{lem: TMAC protocol} is conditional on the number of packets being small. For this reason, unlike for selectors,
the bounds for transmission times into high-degree nodes will be only amortized, not worst case, which is fortunately sufficient to
achieve faster overall running time.

\smallskip

Let $l = \log \aveindegree$ and $h = \log n$. We partition the set of vertices into $h-l+1$ layers, based on their
in-degree:
\begin{align*}
V_l \;&=\; \textstyle \braced{v\in V\suchthat 1 \le \indegree(v) \le 2^l}
\\
V_j \;&=\;\textstyle  \braced{v\in V\suchthat 2^{j-1} < \indegree(v)  \le 2^{j}} \quad \textrm{for $j = l+1,...,h$}
\end{align*}
As in Section~\ref{sec: simpler protocol}, a counting analog of Markov's inequality implies:
\begin{observation}\label{obs: bound on Vj}
$|V_{j}|\le n\aveindegree/2^{j-1}  = 2^{h+l-j+1}$ for $j=l,...,h$.
\end{observation}
Let $f=\onefifth(2h+l)$. The significance of $f$ is that, roughly, for $j \le f-1$
the protocol will use selectors for transmissions into the vertices in  $V_j$, and for $j\ge f$
it will use Protocol~$\TripleSel()$. 

For $j=l,l+1,\dots,f-1$, let
\begin{align*}
\alpha_j \;&=\;	2^{(2h+l-3j)/2}
\\
\LayerSelLatency_j \;&=\; \SelLatency{n}{2^j}{2^j}=\tilde{O}\left(2^{2j}\right)
\end{align*}
and for $j=f,f+1,\dots,h$ let
\begin{align*}
\beta_j  \;&=\;	2^{(2h+l-2j)/3}
\\
\LayerTSLatency_j \;&=\; \TSLatency{n}{2^j}{\beta_j} \;=\; \tilde{O}\left(2^j\beta_j\right)
\end{align*}
where $\TSLatency{n}{k}{r}$ is the running time of $\TripleSelFull{n}{k}{r}$, 
as defined in Section~\ref{sec: auxiliary mac protocol}.
The parameters $\alpha_j$ and $\beta_j$ will be used in the algorithm as threshold values to determine 
flush-eligibility for packets originating from layer $V_j$, for $j<f$ and $j\ge f$, respectively.

\subsection{Condensing Procedures}
\label{subsec: condensing procedures}

We use two condensing procedures. The first procedure, $\InitCondense()$, is executed only once, and it 
will focus on transmitting all rumors out of low-degree layers (those indexed $l,...,{f-1}$).
The second procedure, $\Condense()$, will be executed $h-f+1$ times. Rather than working with individual
rumors, in each of its invocations, $\Condense()$ will be transmitting \emph{packets}, where each
packet represents rumors pending in the same node. 

\smallskip

\myparagraph{Procedure~$\InitCondense()$.}
As in Section~\ref{sec: simpler protocol}, for each $v\in V$ and $j=l,...,h$, let
$\pendingrumors{v}{j}=\pendingrumors{v}{}\cap\rumor(V_j)$ denote the pending rumors in $v$ that originated from $V_j$.
The objective of $\InitCondense()$ is to ensure that each rumor is
either quasi-disseminated or in at least one vertex in
$V_f\cup ...\cup V_h$. The flush-eligibility predicate is:
\[
\eligible'(v)\;\equiv\;
\bigl(\exists j\in\braced{l,...,f-1}\suchthat
|\pendingrumors{v}{j}|\ge\alpha_j\bigr).
\]
\begin{center}
	\framebox{
		\begin{minipage}{6in}
			\begin{tabbing}
				aa \= aa \= aa \= aa \= aa \= aa \= \kill
				\textbf{Procedure}~$\InitCondense()$ \\
				\> \textbf{input:} a rumor $\rumor_v$ for each $v\in V$ \\
				\> \textbf{output:}$\forall v\in V$, $\rumor_v$ is quasi-disseminated or in at least one vertex in $V_f\cup ...\cup V_h$ \\
				\> Let $A_{f-1}=\sum_{j=l}^{f-1}\alpha_j\LayerSelLatency_j$
				\ and \ $F_{f-1}=\sum_{j=l}^{f-1}|V_j|/\alpha_j+1$ \\
				\> Let $T_{f-1}=A_{f-1}+3\dischargetime F_{f-1}$ \\
				\> \textbf{In parallel} execute the following processes for time $T_{f-1}$: \\
				\>\> \textit{Channel $j=l,...,f-1$:} Repeatedly run the $(n,2^j,2^j)$-selector \\
				\>\> \textit{Flush Channel:} Repeatedly execute $\Discharge(\eligible')$ \\
				\> execute $\RoundRobin$ 
			\end{tabbing}
		\end{minipage}
	}
\end{center}

The subscript $f-1$ in the formulas for the running time is meant to represent the phase number,
as Procedure~$\InitCondense()$ will constitute phase~$(f-1)$ of the overall protocol.


\begin{lemma}\label{lem: init condense running time}
	The running time of Procedure~$\InitCondense()$ is
	$T_{f-1}=\tildeO(2^{(6h+3l)/5}) =\tildeO((mn)^{\slthreefifths})$.
\end{lemma}

\begin{proof}
	For $l\le j<f$, by the definitions of $\alpha_j$ and
	$\LayerSelLatency_j$, Observation~\ref{obs: bound on Vj}, and
	$\dischargetime=\tildeO(2^h)$, we have:
\begin{equation*}
	\alpha_j\LayerSelLatency_j=\tildeO(2^{(2h+l+j)/2})
	\qquad\text{and}\qquad
	\dischargetime |V_j|/\alpha_j=\tildeO(2^{(2h+l+j)/2}).
\end{equation*}

Since $j<f=(2h+l)/5$, both quantities are of order $\tildeO(2^{(2h+l+f)/2})=\tildeO(2^{(6h+3l)/5})$.
Summing over $O(\log n)$ indices gives $T_{f-1}=\tildeO(2^{(6h+3l)/5})=\tildeO((mn)^{\slthreefifths})$.
The parallel execution has $O(\log n)$ channels, so the sequential running time
satisfies the same upper bound.
\end{proof}


\begin{lemma}\label{lem: init condense correct}
	For each $v\in V$, after the completion of Procedure~$\InitCondense()$,
	rumor $\rumor_v$ is either quasi-disseminated or in at least one
	vertex in $V_f\cup ...\cup V_h$.
\end{lemma}

\begin{proof}
We start with two claims, both being simple modifications of Observations~\ref{obs: simple prot, number of flushes}
and~\ref{obs: simple prot, path traveled} in Section~\ref{sec: simpler protocol}. 
	
\begin{claim}\label{cla: init condense, number of flushes}
If Procedure $\InitCondense()$ executes $F_{f-1}$ flushes then all rumors from
$V_l\cup ...\cup V_f$ will be disseminated.
\end{claim}

To justify this, we argue as follows:
If $\Discharge()$ executes a flush from a vertex $v$ with
$|\pendingrumors{v}{j}|\ge \alpha_j$, for some $j\le f-1$, the number of pending rumors from $V_j$ decreases by $\alpha_j$.
So for this $j$ we can have at most $|V_j|/\alpha_j$ flushes, and the observation follows
by adding these bounds over all $j = l,...,f$.

\begin{claim}\label{cla: init condense, path traveled}
For any vertex $z$, at time $A_{f-1}$ either its
rumor $\rumor_z$ has been flushed, or it has traversed each path
in $V_l\cup ... V_{f-1}$ that starts at $z$ and has 
at most $\alpha_j$ vertices from $V_j$, for each $j = l,...,f$.
\end{claim}

Indeed, by the definition of selectors, once $\rumor_z$ reaches a vertex $x$, 
the time for it to traverse an edge  $(x,y)$, where $y\in V_j$ for $j\le f-1$, is bounded by $\LayerSelLatency_j$. 
So if a path $P$ satisfies the condition in the observation,
summing these costs along this path gives us time at most $A_{f-1}$, proving Observation~\ref{obs: simple prot, path traveled}.

\smallskip

We now resume the proof of the lemma. Fix some arbitrary vertex $v$.
To show correctness we need to argue that after $T_{f-1}$ steps $v$'s rumor $\rumor_v$
will be either quasi-disseminated or delivered to some node in $V_{f}\cup ... \cup V_h$.

If $\rumor_v$ is already quasi-disseminated by time $A_{f-1}$, we are done. So assume it's not. Then there is
a vertex $w \in V$ that has not received $\rumor_v$ by time $A_{f-1}$. Choose such a $w$ to have the minimum distance from $v$. 
Then some $u\in\inneighbors(w)$ has received $\rumor_v$ by time $A_{f-1}$.
If $w \in V_f \cup ... \cup V_h$, then one more round of $\RoundRobin$ will deliver  $\rumor_v$ to $V_f\cup ... \cup V_h$
by time $A_{f-1}+n \le T_{f-1}$.  Therefore we can assume that $w\in V_l \cup ...\cup V_f$.

The rest of the argument is similar to the proof of correctness of {Protocol}~{\SimpleSparseQuasiGossip} in Section~\ref{sec: simpler protocol}.
As $\rho_v$ was not flushed, $\rumor_v$ must have traveled some path $P$ from $v$ to $u$
as a result of applying selectors on Channels~$l,...,f-1$.
Since $w$ has not received $\rumor_v$, Claim~\ref{cla: init condense, path traveled} implies that
there is $j$ for which path $Pw$ must have strictly more than $\alpha_j$ vertices from $V_j$.
Then $P$ has at least $\alpha_j$ vertices from $V_j$, and their rumors (which are not yet flushed,
because $\rumor_v$ is not quasi-disseminated)
traveled with $\rho_v$ along $P$. So at time $A_{f-1}$ we have $\rumor(V_j\cap P)\subseteq \pendingrumors{u}{j}$,
where $|\rumor(V_j\cap P)|\ge \alpha_j$, making $u$ flush-eligible.
During the remaining $\dischargetime F_{f-1}$ steps $A_{f-1}+1,...,T_{f-1}$, the protocol executes $F_{f-1}$ discharge operations.
These operations will execute flushes as long as $u$ remains flush-eligible.
If they all execute flushes, Claim~\ref{cla: init condense, number of flushes}
implies that $\rumor_v$ will be disseminated.
If not, $u$ must cease being flush-eligible before time $T_{f-1}$, so
at least one of the rumors from $V_j\cap P$ must have been flushed before this time.
This means that $\rumor_v$ was delivered to a vertex whose rumor was disseminated, completing the proof of correctness.
\end{proof}


\smallskip

\myparagraph{Procedure~$\Condense(s)$.}
As mentioned earlier, in Procedure~$\Condense()$ the nodes will transmit information in chunks called \emph{packets}. 
For the purpose of correctness of Procedure~$\Condense()$ the meaning of packets is not important.
(The intuition is that, in the complete protocol, a packet originating from a node $v \in V_s$ will represent the set of rumors pending in $v$
before the execution of $\Condense(s)$.)
The argument $s$ is an integer in the range $f \le s \le h$.
Initially each vertex $v\in V$ has a \emph{packet} that we denote by $\packet_v$. 
The objective of Procedure~$\Condense(s)$ is this: 
for each $v\in V_{s}$, after the completion of the procedure
its packet $\packet_v$ is either quasi-disseminated or delivered to at least one vertex in $V_{s+1}\cup ... \cup V_h$.
Note that we only specify the outcome for packets from $V_{s}$, but
other packets play a critical role in bounding the running time.

For a set $X\subseteq V$, let $\packet(X)=\braced{\packet_x\suchthat x\in X}$. 
By $\pendingpackets{v}{}$ we denote the dynamic set containing still-pending 
packets received by $v\in V$, 
and for $j=l,\ldots,h$ by $\pendingpackets{v}{j}=\pendingpackets{v}{}\cap \packet(V_j)$ we denote 
the subset of $\pendingpackets{v}{}$ with packets originating from layer $V_j$.

In Procedure~$\Condense(s)$, a vertex $v$ is considered \emph{flush-eligible} if it satisfies the
following flush-eligibility predicate:
\begin{equation*}
\eligible_s(v) \;\equiv\; 
	(\,\exists j \in \braced{l,...,f-1} \suchthat |\pendingpackets{v}{j}|\ge \alpha_j\,)
			\,\vee\,
			(\,\exists j \in \braced{f,...,s} \suchthat |\pendingpackets{v}{j}|\ge \beta_j \,)
\end{equation*}
That is, flush-eligibility is determined using threshold values $\alpha_j$ and $\beta_j$ for the
number of stored pending packets, for low- and high-degree vertices, respectively.
\begin{center}
\framebox{
\begin{minipage}{6in}
\begin{tabbing}
aa \= aa \= aa \= aa \= aa \= aa \= \kill
\textbf{Procedure}~$\Condense(s)$ \\
\> \textbf{input:} a packet $\packet_v$ for each $v\in V$ \\
\> \textbf{output:} $\forall v\in V_{s}$, $\packet_v$ is quasi-disseminated or in at least one vertex in $V_{s+1}\cup ...\cup V_{h}$ \\
\> Let $A_s   \;=\; \textstyle \sum_{j=l}^{f-1}\alpha_j\LayerSelLatency_j  + 2\sum_{j=f}^{s}\beta_j\LayerTSLatency_j$
	\ and \  $F_s \;=\; \textstyle \sum_{j=l}^{f-1} {|V_j|}/{\alpha_j}+\sum_{j=f}^{s} {|V_{j}|}/{\beta_j}$\\
\> Let $T_s \;=\; A_s + 3\dischargetime F_s$ \\
\> \textbf{In parallel} execute the following processes for time $T_s$: \\
\>\> \textit{Channel $j=l,\ldots,f-1$:} Repeatedly run the $(n,2^j,2^j)$-selector \\
\>\> \textit{Channel $j=f,\ldots,s$:} Repeatedly run $\TripleSelFull{n}{2^j}{\beta_s}$ for packets $\packetset = \packet(V_s)$ \\
\>\> \textit{RR Channel:} Repeatedly execute $\RoundRobin$ \\
\>\> \textit{Flush Channel:} Repeatedly execute $\Discharge(\eligible_s)$
\end{tabbing}
\end{minipage}
}
\end{center}

The key challenge is that packets attempting to enter a vertex $y$ of high degrees may get
stuck because of large congestion at $y$.
This delay can be very large,
but this can only happen while Procedure~$\Discharge()$ executes flush operations,
so this delay can be charged to reducing the number of pending packets in the network.
Below we formalize this idea. 


\begin{lemma}\label{lem: condense running time}
For each $s = f,...,h$, the running time of Procedure~$\Condense(s)$ is
$T_s = \tilde{O} (2^{(6h+3l)/5}) \;=\; \tildeO((mn)^{\slthreefifths})$.
\end{lemma}

\begin{proof}
The terms in $T_s$ indexed $j = s,...,f-1$ satisfy $\sum_{j=l}^{f-1}\alpha_j\LayerSelLatency_j
+3\dischargetime\sum_{j=l}^{f-1}|V_j|/\alpha_j \le T_{f-1}$.
By Lemma~\ref{lem: init condense running time}, their sum is $\tildeO(2^{(6h+3l)/5})$. 
Next, consider some $f\le j\le s$. Using the definitions of $\beta_j$'s and $\LayerTSLatency_j$'s, 
Observation~\ref{obs: bound on Vj}, and the estimate $\dischargetime = \tildeO(2^h)$, we obtain that
$\beta_j\LayerTSLatency_j \;=\; \tilde{O}(2^{(4h+2l-j)/3})$
and
$\dischargetime |V_{j}|/\beta_j  \;=\;\tilde{O}(2^{(4h+2l-j)/3})$.
The value of $2^{(4h+2l-j)/3}$ decreases with $j$, so, since $j\ge f$, both $\beta_j\LayerTSLatency_j$ and $3\dischargetime |V_{j}|/\beta_j$
are of order $\tilde{O}(2^{(4h+2l-f)/3}) = \tilde{O}(2^{(6h+3l)/5})$ as well.

The above bounds imply that both terms $A_s$ and $3\dischargetime F_s$ of $T_s$ are of order $\tilde{O} (2^{(6h+3l)/5})$, because each consists of
$O(\log n)$ terms estimated above. This shows that
$T_s = \tilde{O} (2^{(6h+3l)/5}) \;=\; \tildeO((mn)^{\slthreefifths})$.
Procedure~$\Condense(s)$ uses $O(\log n)$ radio channels, so, as explained in Section~\ref{sec: preliminaries},
it can be simulated on a single channel without affecting the $\tilde{O}$-estimate of the running time,
completing the proof of Lemma~\ref{lem: condense running time}.
\end{proof}


\begin{lemma}\label{lem: condense correct}
For each $v\in V_{s}$, 
after the completion of Procedure~$\Condense(s)$, $v$'s packet $\packet_v$ is either quasi-disseminated or
delivered to at least one vertex in $V_{s+1} \cup ... \cup V_h$.
\end{lemma}

\begin{proof}
We divide the execution into \emph{frames} of length $\dischargetime$, each frame consisting of
one execution of $\Discharge()$.
If a frame contains a flush, this frame and the two preceding frames are called \emph{paid}; all other frames are \emph{unpaid}.


\begin{claim}\label{cla: paid frames}
The number of flushes during the execution of Procedure~$\Condense(s)$ is at most $F_s$.
Consequently, the number of paid frames is at most $3F_s$. 
\end{claim}

To justify Claim~\ref{cla: paid frames}, note that
a flush triggered by inequality $|\pendingpackets{v}{j}|\ge \alpha_j$, for $j < f$,
can occur at most ${|V_j|}/{\alpha_j}$ times. Similarly, an inequality $|\pendingpackets{v}{j}|\ge \beta_j$, for $j \ge f$,
can trigger at most ${|V_j|}/{\beta_j}$ flushes. By summing up these bounds,
the number of flushes is bounded by $\sum_{j=l}^{f-1} {|V_j|}/{\alpha_j}+\sum_{j=f}^{s} {|V_{j}|}/{\beta_j} = F_s$.
The definition of paid frames immediately implies the $3F_s$ upper bound on their number,
completing the proof of Claim~\ref{cla: paid frames}.

Claim~\ref{cla: paid frames} gives us a bound of $3\dischargetime F_s$ for the total duration of paid frames.
We still need to show that there are enough time steps within the unpaid frames for 
Procedure~$\Condense(s)$ to complete its objective. 
All time steps within unpaid frames will be also called \emph{unpaid}. 


\begin{claim}\label{cla: unpaid times}
Each $y \in V_f\cup ...\cup V_s$ at any unpaid time step $\tau$ satisfies
$|\bigcup_{z\in\inneighbors(y)}\pendingpackets{z}{s}|<\beta_s$.
\end{claim}

To prove Claim~\ref{cla: unpaid times}, we prove its contra-positive. Suppose that
$|\bigcup_{z\in\inneighbors(y)}\pendingpackets{z}{s}| \ge \beta_s$ for some $y \in V_f\cup ...\cup V_s$
at some time step $\tau$. Let $b$ be the frame containing step $\tau$.
We argue that the protocol would execute a flush in frame $b$, $b+1$ or $b+2$, thus making frame $b$ paid.

If a flush is executed in frames $b$ or $b+1$, we are done, so supposed it's not.
Then, for each $z\in \inneighbors(y)$, the packets in $\pendingpackets{z}{s}$ at time $\tau$ will
remain pending until the beginning of frame $b+2$, so during this time the cardinality of $\pendingpackets{z}{s}$  can only get larger. 
As frame $b+1$ will execute $\RoundRobin$ (at least once), 
by the beginning of frame $b+2$ all packets from $\bigcup_{z\in \inneighbors(y)}\pendingpackets{z}{s}$
will reach $y$, so at that time we will have $|\pendingpackets{y}{s}| \ge \beta_s$, making $y$ flush-eligible.
Since there is at least one flush-eligible vertex (namely $y$) at the beginning of frame $b+2$,
frame $b+2$ will execute a flush. This completes the proof of Claim~\ref{cla: unpaid times}.


\begin{claim}\label{cla: traversing edge (x,y)}
Consider an edge $(x,y)$ and suppose that some packet $\packet_v$, for $v\in V_s$, is in $x$ at some time step $\tau$,
and let $y \in V_j$ where $j\le s$. Then
{(i)} If $j < f$ then after time $\LayerSelLatency_j$ packet $\packet_v$ will either get flushed or reach $y$.
{(ii)} If $j\ge f$ then after at most $2\LayerTSLatency_j$ unpaid time steps 
			packet $\packet_v$ will either get flushed or reach $y$.
\end{claim}

Part~(i) of Claim~\ref{cla: traversing edge (x,y)} follows directly from the definition of selectors, so we
focus on the proof of Part~(ii).
Let $\tau'$ be the first time step after $\tau$ for which consecutive steps $\tau',\tau'+1,...,\tau'+\LayerTSLatency_j$
are unpaid. (If $\tau$ is unpaid and at least $\LayerTSLatency_j$ steps before the end of its frame, then $\tau' = \tau$.
Otherwise, $\tau'$ is the beginning of the first unpaid frame after $\tau$.)
Then there are at most $2\LayerTSLatency_j$ unpaid time steps between $\tau$ and $\tau'+\LayerTSLatency_j$.
If $\packet_v$ gets flushed before time $\tau'+\LayerTSLatency_j$, we are done.
Otherwise, since all steps $\tau',\tau'+1,...,\tau'+\LayerTSLatency_j$ are unpaid, Claim~\ref{cla: unpaid times}
implies that at each of these steps we have $|\bigcup_{z\in \inneighbors(y)}\pendingpackets{z}{s}| < \beta_s$.
Further, 
this time interval contains a complete execution of $\TripleSelFull{n}{2^j}{\beta_s}$.
Therefore during this time interval at least one copy of $\packet_v$ will reach $y$, 
completing the proof of Claim~\ref{cla: traversing edge (x,y)}(ii).

\smallskip

Fix some vertex $v\in V_{s}$ from the lemma.  Call a path $Q$ starting at $v$ \emph{$s$-bounded} if 
(i) all vertices of $Q$ belong to $V_l\cup ... \cup V_s$,
(ii) for each $j = l,...,f-1$ it contains at most $\alpha_j$ vertices from $V_j$ , and 
(iii) for each $j = f,...,s$ it contains at most $\beta_j$ vertices from $V_j$. 


\begin{claim}\label{cla: s-bounded paths}
If a path $Q$ starting at a vertex $v\in V_s$ is $s$-bounded then after at most $A_s$ unpaid steps
packet $\packet_v$ will either get flushed or it will reach all vertices on $Q$.
\end{claim}

Claim~\ref{cla: s-bounded paths} follows directly from Claim~\ref{cla: traversing edge (x,y)}:
for each edge $(x,y)$ on $Q$, the unpaid delay 
(the number of unpaid time steps) between the times when $\packet_v$ reaches $x$ and when it reaches $y$
is bounded by $\LayerSelLatency_j$ for $j = l,...,f-1$ and by $2\LayerTSLatency_j$ for $j = f,...,s$.
(Providing that it does not get flushed earlier.) Summing these bounds, the traversal of $Q$ will take time at 
most $\sum_{j=l}^{f-1}\alpha_j\LayerSelLatency_j  + 2\sum_{j=f}^{s}\beta_j\LayerTSLatency_j = A_s$.

\smallskip

To complete the proof of Lemma~\ref{lem: condense correct}, we now argue as follows. If Procedure~$\Condense(s)$ executes $F_s$
flushes then \emph{all} packets will be quasi-disseminated, and we are done. So assume
that the number of flushes is at most $F_s-1$.
Also, let $\tau$ be the first time step such that the interval $[0,\tau]$ contains $A_s$ 
unpaid steps (that is, it consists of exactly $A_s/\dischargetime$ unpaid frames). 
There are at most $3(F_s-1)$ paid frames, each of length
$\dischargetime$. Therefore $\tau \le A_s+3\dischargetime(F_s-1) \le T_s-\dischargetime$.

If $\packet_v$ is already quasi-disseminated by time $\tau$, we are done. So assume it's not. Then there is
a vertex $w \in V$ that has not received $\packet_v$ by time $\tau$. Choose such a $w$ to have the minimum distance from $v$. 
Then some $u\in\inneighbors(w)$ has received $\packet_v$ by time $\tau$.
If $w \in V_{s+1}\cup ... \cup V_h$, then one more round of $\RoundRobin$ will deliver  $\packet_v$ to $V_{s+1}\cup ... \cup V_h$
by time $\tau+n\le T_s-\dischargetime + n \le T_s$.  
(Without loss of generality, $\dischargetime \ge n$.)
Therefore we can assume that $w\in V_l \cup ...\cup V_s$.

Let $P$ be the path from $v$ to $u$ traveled by a particular copy of $\packet_v$ in $u$.
(Break ties arbitrarily.)
Since $\packet_v$ is not yet quasi-disseminated, none of the packets from $P$ have been flushed.
If $P$ contains a vertex in $V_{s+1}\cup ...\cup V_h$  then $\packet_v$ already satisfies the lemma.
Therefore we can assume that all vertices of $Pw$ are in $V_l \cup ...\cup V_s$.

Since $w$ has not received $\packet_v$ by time $\tau$, Claim~\ref{cla: s-bounded paths} implies that
path $Pw$ is \textit{not} $s$-bounded. So there is an index $j\in \braced{l,...,s}$ for which
$Pw$ has either strictly more than $\alpha_j$ vertices from $V_j$, if $j < f$, or
strictly more than $\beta_j$ vertices from $V_j$, if $j\ge f$.
Then, if $j < f$ then $P$ has at least $\alpha_j$ vertices from $V_j$, and if
$j\ge f$ then $P$ has at least $\beta_j$ vertices from $V_j$.
The packets from $V_j\cap P$ have not been flushed before time $\tau$, and (importantly) they
all traveled with $\packet_v$ along $P$ so they are all in $\pendingpackets{u}{j}$ at time $\tau$,
making $u$ flush-eligible at time $\tau$.
Further, $u$ will remain flush-eligible as long as none of the packets from $V_j\cap P$ gets flushed.
Thus, as long as none of the packets from $V_j\cap P$ gets flushed, each frame will
execute a flush operation (not necessarily for $u$).
The interval $[0,\tau]$ has $(\tau-A_s)/\dischargetime$ paid frames and $(T_s-\tau)/\dischargetime$ frames remain after time $\tau$,
for the total of $(T_s - A_s)/\dischargetime = 3F_s$ frames, so, by Claim~\ref{cla: paid frames}, not all the remaining
frames can be paid. This implies that $u$ must eventually cease to be flush-eligible, so 
one of the packets from $V_j\cap P$
must get flushed before time $T_s$, and therefore $\packet_v$ will be quasi-disseminated when Procedure~$\Condense(s)$ ends.
\end{proof}

\subsection{Main Protocol}
\label{subsec: main protocol}

Our main quasi-gossiping protocol is divided into phases. For convenience, we number these phases
$f\!-\!1,f,...,h$. In Phase~$f-1$ we execute Procedure~$\InitCondense()$. After this phase, each rumor 
will either get quasi-disseminated or will end up in a vertex of degree at least $2^{f-1}$, that is in $V_{f-1}\cup ... V_h$.
In the remaining phases $s = f,...,h$,
we use Procedure~$\Condense(s)$ to move rumors (grouped into packets) that are not yet quasi-disseminated to higher-degree
layers.

As in Section~\ref{sec: simpler protocol}, for each vertex $v\in V$ and each $j=l,...,h$, by
$\pendingrumors{v}{j} = \pendingrumors{v}{} \cap \rho(V_j)$ we denote the set 
of pending rumors in $v$ that originated from vertices in layer $V_j$.
At the beginning of Phase~$s=f,...,h$, each vertex $v\in V_{s}$ will pack the rumors from $\pendingrumors{v}{}$
into a fresh Phase-$s$ packet $\packet_v$, while all other vertices will create empty packets.
Procedure~$\Condense(s)$ will then either quasi-disseminate packets 
from $V_{s}$ or move them to $V_{s+1}\cup ...\cup V_h$.

\begin{center}
	\framebox{
		\begin{minipage}{5in}
			\begin{tabbing}
				aa \= aa \= aa \= aa \= aa \= aa \= \kill
				\textbf{Protocol}~\SparseQuasiGossip \\
				\> \textbf{input:} a rumor $\rumor_v$ for each $v\in V$ \\
				\> \textbf{output:} for each $v\in V$, rumor $\rumor_v$ is quasi-disseminated \\
				\> \emph{Phase~$f-1$:} \\
				\>\> Execute $\InitCondense()$ \\
				\> \emph{Phase~$s=f,...,h$ (sequentially):} \\
				\>\> If $v\in V_s$, let $\packet_v=\pendingrumors{v}{}$, else $\packet_v=\emptyset$ \\
				\>\> Execute $\Condense(s)$
			\end{tabbing}
		\end{minipage}
	}
\end{center}


We are now ready to complete the proof of our main result, Theorem~\ref{thm: improved bound sparse graphs}.
By Lemma~\ref{lem: init condense running time}, Phase~$f-1$ takes time $\tildeO((mn)^{\slthreefifths})$.
By Lemma~\ref{lem: condense running time}, the same bound holds for each Phase~$s=f,...,h$.
Since there are $O(\log n)$ phases, the running time of Protocol~{\SparseQuasiGossip} is $\tildeO((mn)^{\slthreefifths})$.
This bound remains valid if we execute the protocol on one radio channel.

The correctness follows quite easily from the correctness of Procedure~$\Condense()$:
By Lemma~\ref{lem: condense correct} and  by induction on $s$, after each Phase~$s = f-1,...,h$, 
all rumors are either quasi-disseminated or are in $V_{s+1} \cup ...\cup V_h$. 
Taking $s=h$, since $V_{h+1} \cup ...\cup V_h = \emptyset$,
when Protocol~{\SparseQuasiGossip} completes all packets will be quasi-disseminated.
This proves correctness of Protocol~{\SparseQuasiGossip} and
completes the proof of Theorem~\ref{thm: improved bound sparse graphs}.



\section{Gossiping in $\Delta$-Regular Graphs}
\label{sec: delta-regular graphs}


In this section we prove Theorem~\ref{thm: Delta-regular graphs}.
To this end, we provide a deterministic protocol that, for each $\Delta$-regular
strongly-connected directed graph, completes gossiping in time $\tildeO(\Delta^\slhalf n)$.
For $\Delta\ge n^{\sltwothirds}$, this running time is achieved by the
general gossiping protocol from~\cite{Gasieniec_etal_det_gossip_04}, so in this
section we assume that $\Delta< n^{\sltwothirds}$.

Fix an $n$-vertex $\Delta$-regular strongly-connected directed graph $G$.
For an integer $d\ge 0$, define the \emph{distance-$d$ out-neighborhood of a vertex $v$},
denoted $\outneighbors_d(v)$, to be the set of vertices reachable from $v$
along paths of length at most $d$. So $\outneighbors_0(v) = \braced{v}$
and $\outneighbors_1(v) = \braced{v}\cup \outneighbors(v)$.
We start with the following lemma.


\begin{lemma}\label{lem: regular, d-outneighbors}
For each vertex $v$ and integer $d\ge 1$,
we have $|\outneighbors_d(v)| \ge  \min(n,\onethird d\Delta)$.
\end{lemma}

\begin{proof}
For $j\ge 0$, let $L_j$ be the set of vertices at distance exactly $j$ from $v$
(the $j$th layer of breadth-first-search). Then
$\outneighbors_d(v) = \bigcup_{j=0}^d L_j$ for all $d\ge 0$.
We claim that if $j\ge 1$ and $L_j\neq\emptyset$ then
\begin{equation}
	|L_{j-1}|+|L_j|+|L_{j+1}| \;\ge\; \Delta.
	\label{eqn: regular, 3 layers}
\end{equation}
The proof of inequality~(\ref{eqn: regular, 3 layers}) is by considering two cases.
If $|L_j\cup L_{j+1}|\ge\Delta$, then~(\ref{eqn: regular, 3 layers}) is trivial.
Otherwise, by the definition of layers, each vertex in $L_j$ has at least $\Delta-|L_j\cup L_{j+1}|$ outgoing edges to $\outneighbors_{j-1}(v)$.
Therefore we can calculate there are at least $(\Delta-|L_j\cup L_{j+1}|)|L_j|$ edges entering $\outneighbors_{j-1}(v)$.
Then the $\Delta$-regularity implies (via the flow-preservation property) that there are at least as many edges leaving $\outneighbors_{j-1}(v)$, and, by the definition of layers, they all go from $L_{j-1}$ to $L_j$.
Therefore some vertex in $L_j$ has at least $\Delta-|L_j\cup L_{j+1}|$ incoming edges from $L_{j-1}$, so
$|L_{j-1}|\ge\Delta-|L_j\cup L_{j+1}|$. This proves~(\ref{eqn: regular, 3 layers}).

\smallskip
It remains to prove that inequality~(\ref{eqn: regular, 3 layers}) implies the lemma.
Let $d_\textmax$ be the smallest integer $d$ for which $\outneighbors_d(v)=V$. (It exists, by strong connectivity.)
For $d\ge d_\textmax$ the lemma is immediate, so it remains to show that
$|\outneighbors_d(v)|\ge\onethird d\Delta$ holds for $1\le d<d_\textmax$.
We prove this by induction on $d$. If $d\in\braced{1,2,3}$, then
$|\outneighbors_d(v)|\ge|\outneighbors_1(v)|=\Delta+1\ge\onethird d\Delta$.
If $d\ge 4$, the induction hypothesis and~(\ref{eqn: regular, 3 layers}),
applied with $j=d-1$, give us:
\[
\begin{aligned}
	|\outneighbors_d(v)|
	&=|\outneighbors_{d-3}(v)|+|L_{d-2}|+|L_{d-1}|+|L_d|\\
	&\ge\onethird(d-3)\Delta+\Delta
	=\onethird d\Delta.
\end{aligned}
\]
This completes the proof.
\end{proof}


\myparagraph{Protocol for $\Delta$-regular graphs.} 
Recall that $\SelLatency{n}{\Delta}{\Delta}=\tildeO(\Delta^2)$
denotes the latency of the $(n,\Delta,\Delta)$-selector.

\begin{center}
\framebox{
\begin{minipage}{3.2in}
\begin{tabbing}
aa \= aa \= aa \= aa \= aa \= aa \= \kill
\textbf{Protocol} $\DeltaRegularQuasiGossip(G)$ \\
 \>  $\beta = n/\Delta^{\slthreehalves}$ \\
 
 \> \textit{Phase~1: spreading rumors} \\
 \>\> run $\beta$ times the $(n,\Delta,\Delta)$-selector \\
 \> \textit{Phase~2: flushing rumors}  \\
 \>\> \textbf{while} {$\exists v\in V : \pendingrumors{v}{}\neq \emptyset$}\ \textbf{do} \\
 \>\>\> find $v\in V$ with maximum load $|\pendingrumors{v}{}|$ \ and \  {$\Flush(v)$}
\end{tabbing}
\end{minipage}
}
\end{center}


\myparagraph{Analysis.}
The correctness of the protocol is trivial, because it will flush all rumors in Phase~2.
As for the running time, Phase~1 takes time $\beta \SelLatency{n}{\Delta}{\Delta} = \tildeO(\Delta^\slhalf n)$.

To estimate the running time of Phase~2,
consider a vertex $u$.  By the definition of selectors and the assumption that the in-degree of each vertex 
is $\Delta$, in Phase~1 $u$'s rumor $\rumor_u$ will traverse any path of length at most $\beta$.
Then Lemma~\ref{lem: regular, d-outneighbors} implies that $u$'s rumor will reach at least
$f=  \min(n,\txtonethird \beta\Delta) \ge \txtonethird n / \Delta^{\slhalf}$ vertices.
The process in Phase~2 is identical to the execution of the well-known greedy set-cover algorithm, where the 
universe is $\braced{\rumor_x\suchthat x\in V}$,
the set family is $\braced{ \pendingrumors{x}{} \suchthat x\in V}$,
and the flush operation for a vertex $v$ in Phase~2 corresponds to adding its set $\pendingrumors{v}{}$ to the set cover.
Since initially each rumor $\rumor_u$ appears in at least $f$ sets $\pendingrumors{v}{}$,
this greedy set cover has size $\tildeO(n/f)$.
Therefore Phase~2 will execute $\tildeO(n/f)$ flush operations, so it will take
time $\tildeO(n^2/f) = \tildeO(\Delta^\slhalf n)$, completing the proof.


\section{Final Comments}
\label{sec: final comments}


While improving the running time of 
$\tildeO(n^\slfourthirds)$ from~\cite{Gasieniec_etal_det_gossip_04} for
general graphs seems difficult, there are some intermediate open problems 
whose solutions could lead to a general algorithm.

Can our protocol be extended to the model where vertex labels are from the
range $[N]$, with $N$ polynomially bounded by $n$ (as in~\cite{Gasieniec_etal_det_gossip_04})
while still being faster than $\tildeO(n^{\slfourthirds})$? 
This would probably require replacing $\RoundRobin$ by some combination of selectors.

Is it possible to solve gossiping in time $\tildeO(\aveindegree n)$, where $\aveindegree= m/n$ is the
average in-degree, at least for $\aveindegree= n^c$ with $c$ very small? That would generalize the
result for graphs where all in-degrees are at most $\aveindegree$~\cite{Gasieniec_etal_det_gossip_04}.

Can the $\TripleSel()$ protocol help in improving time complexity of other problems in ad-hoc
radio networks? One such problem could be information gathering, where rumors from all vertices need to be
collected in a designated target vertex. 
The best known running time for this problem is $\tildeO(n^{1.5})$~\cite{Chrobak_etal_information_gathering_2021}.
(The input graph in this problem may not be strongly connected, so one cannot apply
protocols for gossiping.)


\bigskip
\myparagraph{Acknowledgments.} 
We are very grateful to the anonymous reviewers of the earlier version of this manuscript.
Their detailed comments were very helpful in correcting some mistakes and improving the readability of the
paper.
	

\bibliographystyle{plain}
\bibliography{radio_networks.bib}

@string{soda = " ACM-SIAM Symp. on Discrete Algorithms (SODA)"}

@string{icalp = " Int. Colloq. on Automata, Languages and Programming (ICALP)"}

@string{swat = " Scandinavian Symposium and Workshops on Algorithm Theory (SWAT)"}

@string{cocoon = " Int. Computing and Combinatorics Conference (COCOON)"}

@string{algorithmica = " Algorithmica"}

@article{Czumaj_Rytter_broadcast_06,
author = "Artur Czumaj and Wojciech Rytter",
title = "Broadcasting algorithms in radio networks with unknown topology ",
journal = "Journal of Algorithms ",
volume = "60",
number = "2",
pages = "115 -- 143",
year = "2006"
}

@inproceedings{Clementi_etal_01,
  author    = {Andrea E. F. Clementi and
               Angelo Monti and
               Riccardo Silvestri},
  title     = {Selective families, superimposed codes, and broadcasting
               on unknown radio networks},
  booktitle = {Proc. 12th Annual Symposium on Discrete
               Algorithms (SODA'01)},
  year      = {2001},
  pages     = {709--718}
}

@article{Chrobak_etal_fast_02,
  author    = {Marek Chrobak and
               Leszek Gasieniec and
               Wojciech Rytter},
  title     = {Fast broadcasting and gossiping in radio networks},
  journal   = {Journal of Algorithms},
  volume    = {43},
  number    = {2},
  year      = {2002},
  pages     = {177--189}
}

@article{Bruschi_etal_lower_bound_97,
 author = {Bruschi, Danilo and Del Pinto, Massimiliano},
 title = {Lower Bounds for the Broadcast Problem in Mobile Radio Networks},
 journal = {Distributed Computing},
 issue_date = {April 1997},
 volume = {10},
 number = {3},
 year = {1997},
 issn = {0178-2770},
 pages = {129--135},
 numpages = {7}
}

@article{Chlebus_etal_broadcast_02,
  author    = {Bogdan S. Chlebus and
               Leszek Gasieniec and
               Alan Gibbons and
               Andrzej Pelc and
               Wojciech Rytter},
  title     = {Deterministic broadcasting in ad hoc radio networks},
  journal   = {Distributed Computing},
  volume    = {15},
  number    = {1},
  year      = {2002},
  pages     = {27--38}
}

@article{Kushilevitz_Mansour_lower_bound_98,
  author    = {Eyal Kushilevitz and
               Yishay Mansour},
  title     = {An {$\Omega(D\log (N/D)$} Lower Bound for Broadcast
               in Radio Networks},
  journal   = {SIAM J. Comput.},
  volume    = {27},
  number    = {3},
  year      = {1998},
  pages     = {702--712}
}

@article{alon_etal_lower_bound_91,
  author    = {Noga Alon and
               Amotz Bar-Noy and
               Nathan Linial and
               David Peleg},
  title     = {A Lower Bound for Radio Broadcast},
  journal   = {J. Comput. Syst. Sci.},
  volume    = {43},
  number    = {2},
  year      = {1991},
  pages     = {290--298}
}

@article{Chrobak_etal_randomized_04,
  author    = {Marek Chrobak and
               Leszek Gasieniec and
               Wojciech Rytter},
  title     = {A randomized algorithm for gossiping in radio networks},
  journal   = {Networks},
  volume    = {43},
  number    = {2},
  year      = {2004},
  pages     = {119--124}
}

@inproceedings{Liu_Prabhakaran_randomized_02,
  author    = {Ding Liu and
               Manoj Prabhakaran},
  title     = {On Randomized Broadcasting and Gossiping in Radio Networks},
  booktitle = {Proc. 8th Annual Int. Conference on Computing and Combinatorics (COCOON'02)},
  year      = {2002},
  pages     = {340--349}
}

@inproceedings{Gasieniec_etal_det_gossip_04,
  author    = {Leszek Gasieniec and
               Tomasz Radzik and
               Qin Xin},
  title     = {Faster Deterministic Gossiping in Directed Ad Hoc Radio
               Networks},
  booktitle = {Proc. Scandinavian Workshop
               on Algorithm Theory (SWAT'04)},
  year      = {2004},
  pages     = {397--407}
}

@article{Xu_det_gossip_03,
author={Ying Xu},
title={An {$O(n^{1.5})$} Deterministic Gossiping Algorithm for Radio Networks},
year={2003},
issn={0178-4617},
journal={Algorithmica},
volume={36},
number={1},
pages={93--96},
}

@article{Erdos_etal_families_85,
author={Erd\"{o}s, P. and Frankl, P. and F\"{u}redi, Z.},
title={Families of finite sets in which no set is covered by the union of $r$ others},
year={1985},
issn={0021-2172},
journal={Israel Journal of Mathematics},
volume={51},
number={1--2},
pages={79--89}
}

@inproceedings{DeBonis_etal_selectors_groups_2003,
author= "De Bonis, Annalisa and G{\k{a}}sieniec, Leszek and Vaccaro, Ugo",
title= "Generalized Framework for Selectors with Applications in Optimal Group Testing",
booktitle= "Automata, Languages and Programming",
year= "2003",
publisher= "Springer Berlin Heidelberg",
address= "Berlin, Heidelberg",
pages="81--96"
}

@inproceedings{Christersson_etal_gossiping_bounded_2002,
  author    = {Malin Christersson and
               Leszek Gasieniec and
               Andrzej Lingas},
  title     = {Gossiping with Bounded Size Messages in ad hoc Radio Networks},
  booktitle = {Proc. 29th International Colloquium on Automata, Languages and Programming ({ICALP}'02)},
  pages     = {377--389},
  year      = {2002}
}

@inproceedings{Czumaj_Davies_faster_16,
  author    = {Artur Czumaj and  Peter Davies},
  title     = {Faster Deterministic Communication in Radio Networks},
  booktitle = {Proc. 43rd International Colloquium on Automata, Languages, and Programming (ICALP'16)},
  pages     = {139:1--139:14},
  year      = {2016}
}

@article{DeMarco_distributed_10,
	author = {Gianluca {De Marco}},
	title = {Distributed Broadcast in Unknown Radio Networks},
	journal = {SIAM Journal on Computing},
	volume = {39},
	pages = {2162-2175},
	year = {2010}
}

@incollection{Gasieniec_deterministic_broadcasting_2016,
  author    = {Leszek Gasieniec},
  title     = {Deterministic Broadcasting in Radio Networks},
  booktitle = {Encyclopedia of Algorithms},
	publisher="Springer US",
  pages     = {529--530},
  year      = {2016}
}

@inproceedings{Peleg_broadcasting_review_2007,
  author    = {David Peleg},
  title     = {Time-Efficient Broadcasting in Radio Networks: {A} Review},
  booktitle = {Proc. Distributed Computing and Internet Technology, 4th International Conference,
               {ICDCIT'07}},
  pages     = {1--18},
  year      = {2007}
}

@incollection{Itai_randomized_broadcasting_encyclopedia_2016,
  author    = {Alon Itai},
  title     = {Randomized Broadcasting in Radio Networks},
  booktitle = {Encyclopedia of Algorithms},
  pages     = {1734--1738},
  publisher= {Springer US},
  year      = {2016}
}

@article{komlos_greenberg_non-adaptive_mac_1985,
  author={Komlos, J. and Greenberg, A.},
  journal={IEEE Transactions on Information Theory}, 
  title={An asymptotically fast nonadaptive algorithm for conflict resolution in multiple-access channels}, 
  year={1985},
  volume={31},
  number={2},
  pages={302-306},
  doi={10.1109/TIT.1985.1057020}
}

@inproceedings{kuschner2024permutation,
  author       = {Jordan Kuschner and
                  Yugarshi Shashwat and
                  Sarthak Yadav and
                  Marek Chrobak},
  title        = {On Permutation Selectors and their Applications in Ad-Hoc Radio Networks
                  Protocols},
  booktitle    = {Algorithmics of Wireless Networks - 20th International Symposium {ALGOWIN'24}},
  pages        = {106--116},
  year         = {2024},
  url          = {https://doi.org/10.1007/978-3-031-74580-5\_8},
  doi          = {10.1007/978-3-031-74580-5\_8}
}

@article{Chrobak_etal_information_gathering_2021,
  author    = {Marek Chrobak and Kevin Costello and Leszek Gasieniec},
  title     = {Information Gathering in Ad-Hoc Radio Networks},
  journal   = {Information and Computation},
  year      = {2021},
  volume       = {281},
  pages        = {104769},
  url          = {https://doi.org/10.1016/j.ic.2021.104769},
  doi          = {10.1016/J.IC.2021.104769},
}

@article{Gasieniec_etal_large_labels_2007,
  author       = {Leszek Gasieniec and
                  Aris Pagourtzis and
                  Igor Potapov and
                  Tomasz Radzik},
  title        = {Deterministic Communication in Radio Networks with Large Labels},
  journal      = {Algorithmica},
  volume       = {47},
  number       = {1},
  pages        = {97--117},
  year         = {2007},
  url          = {https://doi.org/10.1007/s00453-006-1212-3},
  doi          = {10.1007/S00453-006-1212-3}
}

@article{gasieniec2002adaptive,
	title={On adaptive deterministic gossiping in ad hoc radio networks},
	author={Gasieniec, Leszek and Lingas, Andrzej},
	journal={Information Processing Letters},
	volume={83},
	number={2},
	pages={89--93},
	year={2002},
	publisher={Elsevier}
}

@article{gkasieniec2007time,
	title={Time efficient centralized gossiping in radio networks},
	author={G{\k{a}}sieniec, Leszek and Potapov, Igor and Xin, Qin},
	journal={Theoretical computer science},
	volume={383},
	number={1},
	pages={45--58},
	year={2007},
	publisher={Elsevier}
}

@article{Clementi_dist_bcast_03,
  author    = {Andrea E. F. Clementi and
               Angelo Monti and
               Riccardo Silvestri},
  title     = {Distributed broadcast in radio networks of unknown topology},
  journal   = {Theor. Comput. Sci.},
  volume    = {302},
  number    = {1-3},
  year      = {2003},
  pages     = {337--364}
}

@article{czumaj_davies_spontaneous_broadcasting_2021,
author = {Czumaj, Artur and Davies, Peter},
title = {Exploiting Spontaneous Transmissions for Broadcasting and Leader Election in Radio Networks},
year = {2021},
issue_date = {April 2021},
publisher = {Association for Computing Machinery},
address = {New York, NY, USA},
volume = {68},
number = {2},
issn = {0004-5411},
url = {https://doi.org/10.1145/3446383},
doi = {10.1145/3446383},
journal = {J. ACM},
month = jan,
articleno = {13},
numpages = {22}
}


\end{document}